\documentclass[twoside]{article}

\usepackage{amsmath,amsthm,amssymb,graphicx}

\usepackage{mathptmx}

\usepackage[text={12.5cm,19cm},centering,paperwidth=17cm,paperheight=24cm]{geometry}

\usepackage[fontsize=10.4pt]{scrextend}

\def\titlerunning#1{\gdef\titrun{#1}}

\makeatletter
\def\author#1{\gdef\autrun{\def\and{\unskip, }#1}\gdef\@author{#1}}
\def\address#1{{\def\and{\\\hspace*{15.6pt}}\renewcommand{\thefootnote}{}\footnote{#1}}\markboth{\autrun}{\titrun}}
\makeatother

\def\email#1{email: \href{mailto:#1}{#1} }
\def\subjclass#1{\par\bigskip\noindent\textbf{Mathematics Subject Classification 2020.} #1}
\def\keywords#1{\par\smallskip\noindent\textbf{Keywords.} #1}

\newenvironment{dedication}{\itshape\center}{\par\medskip}
\newenvironment{acknowledgments}{\bigskip\small\noindent\textit{Acknowledgments.}}{\par}

\newtheorem{thm}{Theorem}[section]

\newtheorem{lem}[thm]{Lemma}

\theoremstyle{definition}

\numberwithin{equation}{section}

\usepackage{esint}
\usepackage{dsfont}
\usepackage[multiple]{footmisc}
\newtheorem{conjecture}[thm]{Conjecture}
\newcommand{\R}{\mathbb{R}}
\newcommand{\Z}{\mathbb{Z}}
\newcommand{\C}{\mathbb{C}}
\newcommand{\N}{\mathbb{N}}
\newcommand\1{{\mathds 1}}
\newcommand{\ii}{\infty}

\newcommand{\nn}{\nonumber}
\DeclareMathOperator{\sn}{sn}
\DeclareMathOperator{\VTr}{\underline{Tr}\,}

\DeclareMathOperator{\tr}{Tr\,}
\renewcommand{\epsilon}{\varepsilon}
\newcommand{\eps}{\varepsilon}

\newcommand{\norm}[1]{ \left| \! \left| #1 \right| \! \right| }

\def\rd{{\mathrm{d}}}
\def\cL{{\mathcal L}}

\usepackage[hyperfootnotes=false,colorlinks=true,allcolors=blue]{hyperref}

\begin{document}

\titlerunning{Periodic Lieb-Thirring inequality}

\title{\textbf{The periodic Lieb-Thirring inequality}}

\author{Rupert L. Frank \and David Gontier \and Mathieu Lewin}

\date{}

\maketitle

\address{R. L. Frank: Mathematics 253-37, Caltech, Pasa\-de\-na, CA 91125, USA, and Mathematisches Institut, Ludwig-Maximilans Universit\"at M\"unchen, Theresienstr. 39, 80333 M\"unchen, Germany, and Munich Center for Quantum Science and Technology (MCQST), Schellingstr. 4, 80799 M\"unchen, Germany; \email{rlfrank@caltech.edu} \and D. Gontier: CEREMADE, University of Paris-Dauphine, PSL University, 75016 Paris, France; \email{gontier@ceremade.dauphine.fr} \and M. Lewin: CNRS and CEREMADE, University of Paris-Dauphine, PSL University, 75016 Paris, France; \email{mathieu.lewin@math.cnrs.fr}}

\begin{dedication}
To Ari Laptev on the occasion of his 70th birthday
\end{dedication}

\begin{abstract}
We discuss the Lieb-Thirring inequality for periodic systems, which has the same optimal constant as the original inequality for finite systems. This allows us to formulate a new conjecture about the value of its best constant. To demonstrate the importance of periodic states, we prove that the 1D Lieb-Thirring inequality at the special exponent $\gamma=3/2$ admits a one-parameter family of periodic optimizers, interpolating between the one-bound state and the uniform potential. Finally, we provide numerical simulations in 2D which support our conjecture that optimizers could be periodic.

\end{abstract}

\bigskip

The Lieb-Thirring inequality~\cite{LieThi-75,LieThi-76} plays an important role in mathematical physics. It has been used to give a short proof of the stability of matter~\cite{LieThi-75,LieSei-09} and it is a fundamental tool for studying large fermionic systems. A famous problem is to determine its \emph{optimal constant}. In this paper we discuss the extension of the Lieb-Thirring inequality to periodic systems, which has the same optimal constant as the original inequality. In light of our recent work~\cite{FraGonLew-20_ppt}, we then formulate a conjecture about the value of this best constant, different from the original Lieb-Thirring conjecture in~\cite{LieThi-76}. In one space dimension we prove the optimality of some periodic potentials in the particular (completely integrable) case where the exponent equals $\gamma=3/2$. Finally, we provide numerical simulations in two space dimensions which show that periodic potentials give a better constant than all previously considered functions. 

\section{The periodic Lieb-Thirring inequality}

Let us first recall the usual Lieb-Thirring inequality for finite systems. Let $d\geq 1$ and 
\begin{equation}
\begin{cases}\gamma\geq\frac12& \text{for $d=1$,}\\ 
\gamma>0& \text{for $d=2$,}\\
\gamma\geq0& \text{for $d\geq3$.}
\end{cases}
\label{eq:constraint_kappa}
\end{equation}
The Lieb-Thirring inequality states that there is a universal constant $L_{\gamma,d}<\ii$ so that 
\begin{equation}
 \sum_{n=1}^\ii |\lambda_n(-\Delta+V)|^\gamma \leq L_{\gamma,d} \int_{\R^d} V(x)_-^{\gamma+\frac{d}2}\,{\rd x}
\label{eq:LT_V}
\end{equation}
for all $V\in L^{\gamma+d/2}(\R^d,\R)$. Here $a_-=\max(0,-a)$ is the negative part and $\lambda_n(-\Delta+V)$ is the $n$th min-max level of $-\Delta+V$ in $L^2(\R^d)$, which equals the $n$th negative eigenvalue (counted with multiplicity) when it exists and 0 otherwise. Under our assumptions on $V$ and $\gamma$, the operator $-\Delta+V$ has a bounded-below quadratic form and we always work with its Friedrichs extension. In the following $L_{\gamma,d}$ will always denote the \emph{best (smallest) constant}. The inequality~\eqref{eq:LT_V} was shown by Lieb-Thirring in~\cite{LieThi-75,LieThi-76} when all the inequalities are strict in~\eqref{eq:constraint_kappa}. The critical cases $\gamma=0$ in $d\geq3$ and $\gamma=1/2$ in $d=1$ are respectively due to  Cwikel-Lieb-Rozenblum~\cite{Cwikel-77,Lieb-76b,Rozenbljum-72} and Weidl~\cite{Weidl-96}. We refer to the recent review~\cite{Frank-20_ppt} for an up-to-date account of what is known and what is open for the inequality~\eqref{eq:LT_V}. 

Next, we turn to (infinite) periodic systems. Let $V\in L^{\gamma+d/2}_{\rm loc}(\R^d)$ be a periodic function, with $\gamma$ as in~\eqref{eq:constraint_kappa}. This means that there exists a discrete subgroup $\cL=\{\sum_{j=1}^dz_jv_j\ :\ z_j\in\Z\}$ with $v_1,...,v_d$ a basis of $\R^d$ such that $V(x+\ell)=V(x)$ for all $\ell\in\cL$ and almost all $x\in\R^d$. We emphasize here that the lattice $\cL$ is not fixed \emph{a priori} but it is part of the definition of $V$. There are in fact several possible such lattices, for instance $k\cL$ with $k\in\N$. Let $C=\{x\in\R^d\ :\ |x|=\min_{\ell\in\cL}|x-\ell|\}$ be the Wigner-Seitz unit cell of $\cL$. It is well known that the integral per unit volume converges to
\begin{equation}
 \lim_{\Omega_n\nearrow\R^d}\frac1{|\Omega_n|}\int_{\Omega_n} V(x)_-^{\gamma+\frac{d}2}\rd x=\frac1{|C|}\int_CV(x)_-^{\gamma+\frac{d}2}\rd x=:\fint V(x)^{\gamma+\frac{d}2}_-\,\rd x
 \label{eq:integral_unit_vol}
\end{equation}
under some natural conditions on the sequence of domains $\Omega_n$. For instance we can just choose $\Omega_n=n\Omega$, where $\Omega$ is a fixed convex open set. Of course, $\fint V(x)^{\gamma+d/2}_-\rd x$  in~\eqref{eq:integral_unit_vol} does not depend on the chosen lattice $\cL$ of periodicity for $V$. 

The spectrum of $-\Delta+V$ is properly described using the Bloch-Floquet transform~\cite[Sec.~XIII.16]{ReeSim4}. Let $\cL^*$ be the dual lattice (such that $k\cdot \ell\in2\pi\Z$ for all $k\in\cL^*$ and $\ell\in\cL$) and $B$ its Wigner-Seitz cell (the Brillouin zone). For any quasimomentum $\xi\in B$, we denote by $H_\xi:=|-i\nabla+\xi|^2+V(x)$ the Schr\"odinger operator on $L^2(C)$ with periodic boundary conditions on $\partial C$. Its eigenvalues form a non-decreasing sequence $\eps^V_n(\xi)$, each of which is Lipschitz in $\xi\in B$. The spectrum of the periodic Schr\"odinger operator is the union of the Bloch bands
$\sigma(-\Delta+V)=\bigcup_{n\geq1}\eps^V_n(B).$
Let us denote by $(-\Delta+V)_{|\Omega}$ the operator $-\Delta+V$ restricted to a domain $\Omega$ with Dirichlet boundary conditions at $\partial\Omega$.
We then have the thermodynamic limits~\cite{ReeSim4,Nakamura-01,DoiIwaMin-01}
\begin{align}
\lim_{n\to\ii}\frac{\sum_{j=1}^\ii\lambda_j\big((-\Delta+V)_{|\Omega_n}\big)_-^\gamma}{|\Omega_n|}&=\lim_{n\to\ii}\frac{\tr\left(\1_{\Omega_n}(-\Delta+V)^\gamma_-\1_{\Omega_n}\right)}{|\Omega_n|}\nn\\
&=\frac1{(2\pi)^d}\sum_{j\geq1}\int_B\eps_j^V(\xi)_-^\gamma\,\rd\xi=:\VTr(-\Delta+V)_-^\gamma.
\label{eq:trace_unit_vol}
\end{align}
The trace per unit volume $\VTr(-\Delta+V)_-^\gamma$ defined in~\eqref{eq:trace_unit_vol} does not depend on the chosen $\cL$, since it equals the first two limits. The following is a simple consequence of the usual Lieb-Thirring inequality but we have not found it stated explicitly anywhere. 

\begin{thm}[Periodic Lieb-Thirring inequality]\label{thm:LT_periodic}
For every periodic function $V\in L^{\gamma+d/2}_{\rm loc}(\R^d,\R)$ with $\gamma$ as in~\eqref{eq:constraint_kappa}, we have
\begin{equation}
\VTr(-\Delta+V)_-^\gamma\leq L_{\gamma,d} \fint V(x)^{\gamma+\frac{d}2}_-\,\rd x
\label{eq:LT_periodic}
\end{equation}
with the same optimal constant $L_{\gamma,d}$ as in the original Lieb-Thirring inequality~\eqref{eq:LT_V}. 
\end{thm}

\begin{proof}
We quickly outline the argument. Let $V$ be a periodic function as in the statement. From the variational principle and the Lieb-Thirring inequality~\eqref{eq:LT_V}, we  have 
$$\sum_{j\geq1}\lambda_j\left((-\Delta+V)_{|\Omega}\right)^\gamma_-\leq \sum_{j\geq1}\lambda_j\left(-\Delta+V\1_{|\Omega}\right)^\gamma_-\leq L_{\gamma,d}\int_\Omega V(x)^{\gamma+d/2}_-\rd x.$$
Passing then to the limit $\Omega\nearrow\R^d$ using~\eqref{eq:integral_unit_vol} and~\eqref{eq:trace_unit_vol}, we obtain~\eqref{eq:LT_periodic}. Let us now call $L_{\gamma,d}^{\rm per}\leq L_{\gamma,d}$ the best constant in~\eqref{eq:LT_periodic} and show that it coincides with $L_{\gamma,d}$. Let $v\in C^\ii_c(\R^d,\R)$ and define the $\ell\Z^d$--periodic function $V_\ell(x)=\sum_{z\in\Z^d}v(x+\ell z)$. The negative spectrum of the periodic operator $-\Delta+V_\ell$ is composed of very narrow bands about each eigenvalue $\lambda_n(-\Delta+v)$~\cite{Daumer-93}. Passing to the limit $\ell\to\ii$ provides the reverse inequality
$\sum_{j\geq1}\lambda_j(-\Delta+v)_-^\gamma\leq L_{\gamma,d}^{\rm per}\int_{\R^d}v^{\gamma+d/2}_-$.
Optimizing over $v$ and using a simple density argument, we find that $L_{\gamma,d}^{\rm per}=L_{\gamma,d}$.
\end{proof}

We have considered the periodic case because such systems usually play an important role in statistical mechanics~\cite{BlaLew-15}. But a result similar to Theorem~\ref{thm:LT_periodic} holds for any potential $V$ which has some kind of ergodicity allowing for the limits~\eqref{eq:integral_unit_vol} and~\eqref{eq:trace_unit_vol} to exist and be independent of $\Omega_n$. 

Next, we discuss our conjecture on $L_{\gamma,d}$. Let us define the best constant 
\begin{equation}
L_{\gamma,d}^{(N)}:=\sup_{V\in L^{\gamma+d/2}(\R^d)}\frac{ \sum_{n=1}^N |\lambda_n(-\Delta+V)|^\gamma}{\int_{\R^d} V_-^{\gamma+\frac{d}2}}
\label{eq:LT_V_N}
\end{equation}
for the inequality similar to~\eqref{eq:LT_V} when only the $N$ first eigenvalues are retained. The one-bound-state constant $L_{\gamma,d}^{(1)}$ can be expressed in terms of the constant for the Gagliardo-Nirenberg embedding $H^1(\R^d)\hookrightarrow L^p(\R^d)$ with $p=(2\gamma+d)/(\gamma+d/2-1)$~\cite{Frank-20_ppt}. Another important number is the semiclassical constant
\begin{equation}
L_{\gamma,d}^{\rm sc}:=(2\pi)^{-d}\int_{\R^d}(|p|^2-1)_-^\gamma\rd p=
\dfrac{\Gamma(\gamma + 1)}{2^d \pi^{d/2} \Gamma(\gamma + d/2 + 1)}.
\label{eq:LT_sc}
\end{equation}
We have $L_{\gamma,d}\geq L^{\rm sc}_{\gamma,d}$ due to the semiclassical limit
$$\lim_{\hbar\to0}\frac{\tr\,(-\Delta+V(\hbar\cdot))_-^\gamma}{\int_{\R^d}V(\hbar x)^{\gamma+\frac{d}2}_-\,\rd x}=\lim_{\hbar\to0}\frac{\hbar^d\,\tr\,(-\hbar^2\Delta+V)_-^\gamma}{\int_{\R^d}V(x)^{\gamma+\frac{d}2}_-\,\rd x}=L^{\rm sc}_{\gamma,d}$$
for any $V\in C^0_c(\R^d,\R)$ with $V_-\neq0$. A different way to understand $L^{\rm sc}_{\gamma,d}$ is to simply take a constant potential $V=-\mu<0$ in the periodic inequality~\eqref{eq:LT_periodic} and remark that 
$$\fint V(x)^{\gamma+\frac{d}2}_-\,\rd x=\mu_+^{\gamma+d/2},\quad  \VTr(-\Delta-\mu)_-^\gamma=(2\pi)^{-d}\int_{\R^d}(|p|^2-\mu)_-^\gamma\,\rd p=\mu_+^{\gamma+d/2}L_{\gamma,d}^{\rm sc}.$$
The $N$-bound state constant $L_{\gamma,d}^{(N)}$ also has an interpretation in the setting of periodic systems. As we saw in the proof of Theorem~\ref{thm:LT_periodic}, a potential $v\in L^{\gamma+d/2}(\R^d)$ with exactly $N$ negative eigenvalues can be turned into a periodic potential $V_\ell(x)=\sum_{z\in\Z^d}v(x+\ell z)$ of large periodicity $\ell$ with $N$ negative Bloch bands converging to the eigenvalues of $-\Delta+v$ in the limit $\ell\to\ii$. We see that although the periodic model is perfect to describe the infinite uniform state, finite states can only be recovered asymptotically for an infinite period. For the usual Lieb-Thirring inequality, the reverse holds: finite bound-state potentials are well described but the semiclassical constant can only be obtained by a limiting procedure. Depending on what we think the optimal constant will be, it seems appropriate to study either~\eqref{eq:LT_V} or~\eqref{eq:LT_periodic}.

The original Lieb-Thirring conjecture~\cite{LieThi-76} stated that the optimal constant is given either by the one-bound state or by the semiclassical case: 
$L_{\gamma,d}\overset{?}{=}\max(L_{\gamma,d}^{(1)},L_{\gamma,d}^{\rm sc})$. Rephrased for our periodic inequality~\eqref{eq:LT_periodic}, this would mean that either $V\equiv\text{cnst}$ is optimal, or there is no optimizer since those would have an infinite period. In dimensions $d\leq7$ the two curves $\gamma\mapsto (L^{(1)}_{\gamma,d},L^{\rm sc}_{\gamma,d})$ cross at a unique critical exponent $\gamma_{1\cap\rm sc}(d)$ equal to $3/2$ in 1D, and approximately equal to $1.165378$ and $0.862689$ in 2D and 3D respectively~\cite{FraGonLew-20_ppt}. In dimensions $d\geq8$ we always have $L^{(1)}_{\gamma,d}<L^{\rm sc}_{\gamma,d}$. It is known that $L_{\gamma,d}=L_{\gamma,d}^{\rm sc}$ for all $\gamma\geq3/2$ in all dimensions~\cite{LieThi-76,AizLie-78,LapWei-00}, and that $L_{1/2,1}=L_{1/2,1}^{(1)}$~\cite{HunLieTho-98}. But the exact value of $L_{\gamma,d}$ has not been found in all the other cases. 

It is now understood that the situation ought to be more complicated than what was hoped in~\cite{LieThi-76}, except probably for $d=1$. In~\cite{FraGonLew-20_ppt} we proved that 
\begin{equation}
L^{(1)}_{\gamma,d}<L^{(2)}_{\gamma,d}\leq L_{\gamma,d}\quad\text{for }\gamma >\max\left\{0,2-\frac{d}2\right\},
 \label{eq:result_FraGonLew}
\end{equation}
\begin{equation}
 L^{(N)}_{\gamma,d}< L_{\gamma,d}\quad\text{for all $N\geq1$ when }    
\gamma\begin{cases}
     >\frac32&\text{for $d=1$,}\\
    >1&\text{for $d=2$,}\\
    \geq1&\text{for $d\geq3$.}
    \end{cases}
 \label{eq:result_FraGonLew2}
\end{equation}
This showed that the one-bound state constant $L^{(1)}_{\gamma,d}$ cannot be optimal in regions where it was known that $L^{(1)}_{\gamma,d}>L^{\rm sc}_{\gamma,d}$, like  $1/2<\gamma<0.862689$ in 3D and  $1<\gamma<1.165378$ in 2D. In fact, no finite-bound-state case can be optimal when $\gamma$ satisfies~\eqref{eq:result_FraGonLew2}. 

In~\cite{FraGonLew-20_ppt} we mentioned the possibility of a different scenario for the Lieb-Thirring optimal constant, which we would like to detail here. For $\gamma>1$ it is possible to interpret the Lieb-Thirring problem as optimizing the state of a quantum system described by its one-particle density matrix $\Gamma$, in the presence of a local nonlinear attraction of the form $-\int_{\R^d}\Gamma(x,x)^p\,\rd x$ and with a Tsallis-type entropy $\tr(\Gamma^q)$ where $p=(\gamma+d/2)'$ and $q=\gamma'$ are the corresponding dual exponents. This is explained later in Appendix~\ref{app:entropy} for completeness. In this physical interpretation, the property $L^{(N)}_{\gamma,d}< L_{\gamma,d}$ for all $N\geq1$ means that the system is willing to form an infinite cluster of particles. Stable infinite systems are usually found in several possible phases depending on the value of the parameters, including fluids and solids~\cite{BlaLew-15}. In our situation a fluid corresponds to $V$ being constant, in which case we obtain $L^{\rm sc}_{\gamma,d}$, as we have seen. The possibility of having a solid phase where $V$ is periodic does not seem to have been considered before in the literature for the best Lieb-Thirring constant. More complicated phases are sometimes observed in statistical mechanics (for instance translation-invariance is rarely broken in 2D but rotation-invariance can be). Here we are in a mean-field setting where $V\equiv\text{cnst}$  is the only translation-invariant state. This leads us to the following 

\begin{conjecture}[Value of the Lieb-Thirring constant]\label{conjecture}
\it For all $d\geq1$ and $\gamma$ satisfying~\eqref{eq:constraint_kappa}, we have that 


\noindent$\bullet$ \textbf{either} there exists $N\in\N$ and a potential $V\in L^{\gamma+d/2}(\R^d)$ with exactly $N$ negative eigenvalues optimizing~\eqref{eq:LT_V}, so that $L_{\gamma,d}=L^{(N)}_{\gamma,d}$;


\noindent$\bullet$ \textbf{or} $L^{(N)}_{\gamma,d}< L_{\gamma,d}$ for all $N\geq1$ and there exists an optimal \textbf{periodic} potential $V\in L^{\gamma+d/2}_{\rm loc}(\R^d)$ optimizing~\eqref{eq:LT_periodic}. This potential can be constant (then $L_{\gamma,d}=L^{\rm sc}_{\gamma,d}$) or not. 
\end{conjecture}

The Lieb-Thirring inequalities~\eqref{eq:LT_V} and~\eqref{eq:LT_periodic} are invariant under the scaling $t^2V(tx)$ and optimizers are never unique. In the periodic case all periods are therefore possible by scaling. The inequalities are also invariant under space translations and the possibility of a (non-constant) periodic optimizer would be a breaking of this symmetry. 

Our conjecture has been proved for all $\gamma\geq3/2$ where the optimal potential is constant, and for $\gamma=1/2$ in dimension $d=1$ where the one-bound-state case is best. A more precise conjecture would be that the system is in a fluid phase ($V\equiv\text{cnst}$) for $\gamma$ larger than some critical $\gamma_{\rm sc}(d)$, then goes to a solid phase when we decrease $\gamma$ until it hits a point at which the period diverges and finite systems become better. 

In dimension $d=1$, numerics indicates that one should have $L_{\gamma,1}=L_{\gamma,1}^{(1)}$ for $\gamma\leq3/2$ and the fluid phase $L_{\gamma,1}=L_{\gamma,1}^{\rm sc}$ for $\gamma\geq3/2$, see~\cite{Levitt-14} and Section~\ref{sec:numerics} below. We will see in the next section that the solid phase actually occurs, but only at the special point $\gamma=3/2$. In some sense all the possible phase transitions seem to be compressed at the unique point $\gamma=3/2$. In dimension $d=2$, the solid phase could be optimal in the region $1<\gamma<\gamma_{\rm sc}(2)$ for some $\gamma_{\rm sc}(2)\in(1.165378,3/2]$, see Section~\ref{sec:numerics}. In dimension $d=3$, it could occur for $1/2<\gamma<1$, if we believe that  
$L_{\gamma,3}>L_{\gamma,3}^{(N)}$ for all $N\geq1$ instead of just $N=1$ when $\gamma>1/2$, as well as the Lieb-Thirring conjecture that the semiclassical constant becomes optimal at $\gamma_{\rm sc}(3)=1$. 

\section{The one-dimensional integrable case $\gamma=3/2$ }\label{sec:1D}

We provide here a new result for $\gamma=3/2$ in dimension $d=1$. Using a link with the Korteweg-de Vries (KdV) equation, it was proved by Lieb and Thirring in~\cite{LieThi-76} that 
\begin{equation}
L_{3/2,1}= L^{(N)}_{3/2,1} =L^{\rm sc}_{3/2,1}=\frac{3}{16},\qquad \forall N\in\N.
\label{eq:value_L_1D}
\end{equation}
In fact, $L^{(N)}_{3/2,1} $ is attained for every $N\in\N$ and thus the Lieb-Thirring inequality has infinitely many optimizers at $\gamma=3/2$, modulo space translations. We show that it also admits a continuous family of \emph{periodic} optimizers, parametrized by their period. 

\begin{thm}[Periodic optimizers in the integrable case]\label{thm:1D}
Let $\gamma=3/2$ and $d=1$. For all $0<k<1$, we have equality in~\eqref{eq:LT_periodic} for the periodic Lam\'e potential
$$V_k(x) = 2k^2 \sn(x|k)^2 - 1 - k^2$$
of minimal period $\ell=2K(k)>0$. Here $\sn(\cdot|k)$ is a Jacobi elliptic function with modulus $k$ and $K(k)$ is the complete elliptic integral of the first kind with modulus $k$.
\end{thm}

It is known that $V_k(x) \to -1$ uniformly as $k\to 0$ (resp.~$\ell\to0$) and that $V_k(x)\to -2(\cosh x)^{-2}=:V_1(x)$ locally uniformly as $k\to 1$ (resp.~$\ell\to\ii$). See Figure~\ref{fig:periodic_family}. Since $V_1$ is the optimum for $L^{(1)}_{3/2,1}$, $V_k$ interpolates continuously between the semiclassical and the one-particle regimes when we vary the periodicity $\ell$. The potential $V_k$ is a periodic traveling (cnoidal) wave for KdV~\cite{KorVri-95} and it is in fact a periodic superposition of $V_1$~\cite{Toda-70}. It is also very well known in the theory of one-dimensional periodic Schr\"odinger operators, since it produces a unique negative Bloch band and only one gap~\cite{Ince-40b,MagWin-66,Sutherland-73}. There are explicit families of periodic potentials with exactly $K$ negative Bloch bands for any $K\geq1$~\cite{DubNov-74} but they will not be discussed here.

\begin{figure}[!t]
\centering
\begin{tabular}{ccc}
\includegraphics[width=3.8cm]{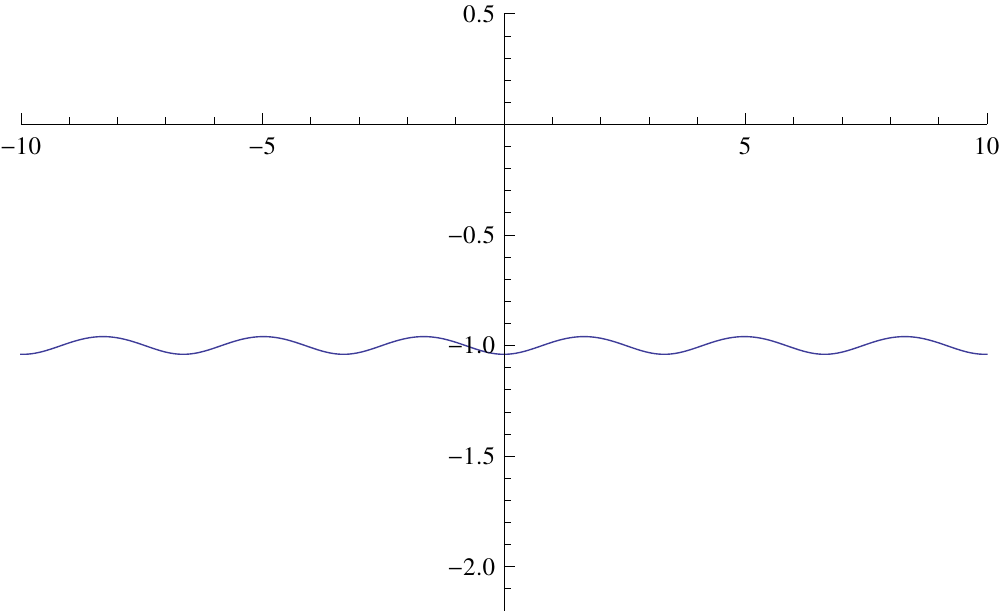}&\includegraphics[width=3.8cm]{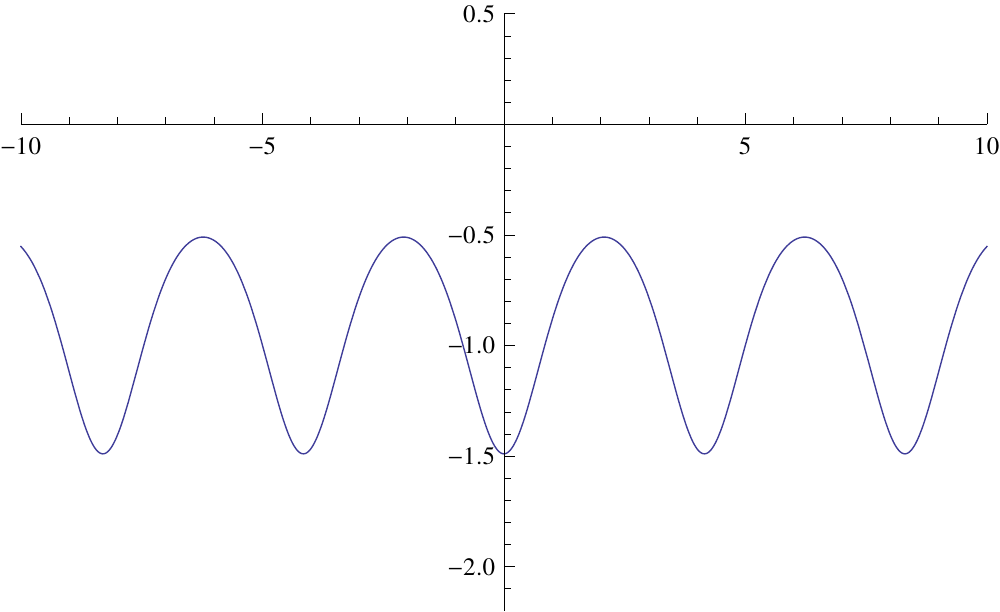}&\includegraphics[width=3.8cm]{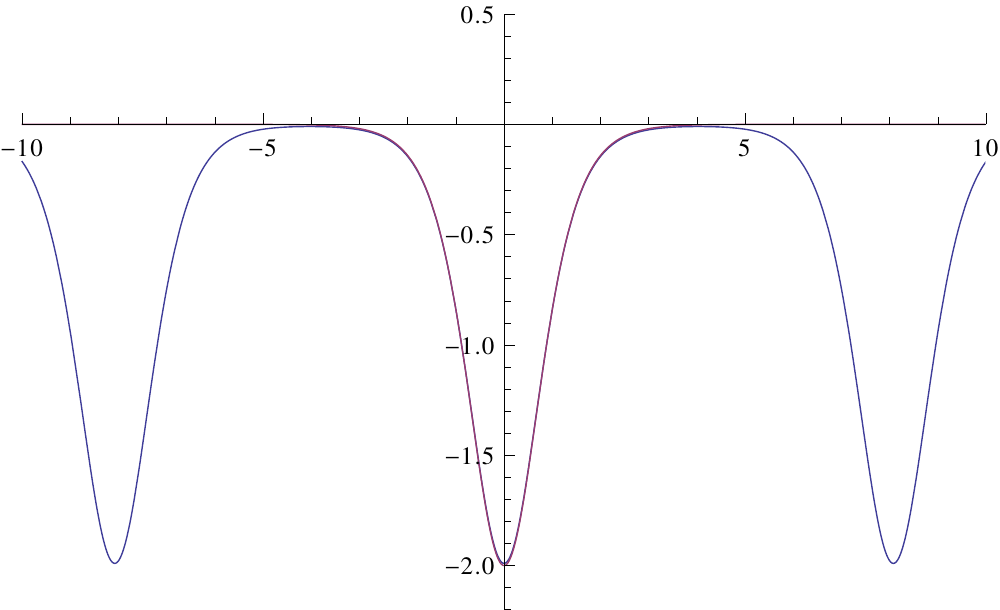}\\
\scriptsize $k=0.2$ ($\ell=3.32$)&\scriptsize $k=0.7$ ($\ell=4.15$)& \scriptsize $k=0.995$ ($\ell=8.08$)
\end{tabular}

\caption{Plot of $V_k$ for different values of $k$ and the period $\ell=2K(k)$ (together with $V_1$ for the last). \label{fig:periodic_family}}
\end{figure}

\begin{proof}
	We fix $0<k<1$. It is well known that $V_k$ is $2K(k)$ periodic~\cite{GraRyz}. Recall that the density of states $n(E)$ is defined by $\VTr \1_{(-\ii,E]}(-\Delta+V_k)=\int_{-\ii}^En(E')\,dE'$. The proof is based on the following explicit formula:
\begin{equation}
	n(E)=\frac{E+c}{2\pi\sqrt{(E+1)(E+k^2)E}}\left(-\1_{(-1,-k^2)}(E)+\1_{(0,+\ii)}(E)\right),
 \label{eq:nE}
\end{equation}
	where $c = \frac{k^2}{2K(k)} \int_{-K(k)}^{K(k)} \sn(x|k)^2\,dx$.\footnote{We do not need to compute $c$ explicitly, although this could be done using~\cite[5.134]{GraRyz}.} The formula~\eqref{eq:nE} is probably known to experts, but we provide a complete derivation in Appendix~\ref{app:details} for the convenience of the reader. This formula implies that the spectrum has a unique isolated Bloch band: $\sigma(-\Delta+V_k)=[-1,-k^2]\cup[0,\ii)$, as mentioned previously. Using 
	$$
	\int_{k^2}^1 \frac{t\,dt}{\sqrt{(1-t)(t-k^2)}} = \frac{\pi}{2} \left( 1 + k^2 \right)
	\quad\text{and}\quad
	\int_{k^2}^1 \frac{t^2\, dt}{\sqrt{(1-t)(t-k^2)}} = \frac{\pi}{8} \left( 3 +2 k^2 + 3 k^4 \right)
	$$
	(which follow from standard beta function integrals letting $t=(1-a)s+a$),
	we obtain
	$$
\VTr(-\Delta+V)_-^{3/2}=\int_{-\infty}^0 n(E)\, E_-^\frac32\,dE 
= \frac{1}{16} \left( 3 +2 k^2 + 3 k^4 \right) - \frac{c}{4}\left( 1 + k^2 \right).
	$$
	To compute the right side of~\eqref{eq:LT_periodic}, we use~\cite[5.131]{GraRyz}
	$$
	\int \,du = \sn u \ {\rm cn }\,u \ {\rm dn }\, u + 2 (1+k^2) \int \sn^2 u \,du - 3 k^2 \int \sn^4 u\,du \,,
	$$
	where we drop the parameter $k$ from the notation for simplicity. Since $\sn 0 = \sn 2K = 0$ by~\cite[8.151]{GraRyz}, we infer that
    $$
         \int_0^{2K} \sn^4 u\,du = \frac{1}{3k^2} \left(  2 (1+k^2) \int_0^{2K} \sn^2 u\,du - 2K \right) = 
         \frac{2 K}{3k^2} \left(  2 (1+k^2) \frac{c}{k^2} - 1 \right).
    $$
	Using $L_{3/2,1}=3/16$ as we have recalled in~\eqref{eq:value_L_1D}, this implies
	\begin{align*}
	 \frac1\ell \int_{-\frac\ell 2}^{\frac\ell 2} V(x)_-^2\,dx  & =   \frac{1}{2K} \left( 4k^4 \int_0^{2K} \sn^4 u\,du - 4k^2(1+k^2) \int_0^{2K} \sn^2 u\,du + (1+k^2)^2 2K \right) \\
	& =  \frac{4k^2}{3} \left( 2(1+k^2) \frac{c}{k^2} - 1 \right) - 4(1+k^2)c  + (1+k^2)^2  \\
	& = \big(L_{\frac32,1}\big)^{-1}\left(\frac{1}{16} \left( 3 +2 k^2 + 3 k^4 \right) - \frac{c}{4}\left( 1 + k^2 \right)\right)
	\end{align*}
and proves the assertion. 
\end{proof}

\section{Numerical simulations in 1D and 2D}\label{sec:numerics}

We have computed a numerical approximation of the optimal periodic potential $V$ in dimensions $d=1$ and $d=2$. In order to remove the scaling invariance, we fix a lattice $\cL$ with unit cell $C$ of volume $|C|=1$ and Brillouin zone $B$, and work with the additional constraint that the norm of $V$ is fixed. We also retain a fixed number $K$ of Bloch bands. In other words, for $I >0$ and $K\in\N$ we set
\begin{equation} \label{eq:def:M}
 L_{\gamma,d,\cL}(K,I) := \sup \bigg\{ \sum_{n=1}^K\frac1{|B|}\int_{B}  \eps^V_{n}(\xi)_-^\gamma \rd \xi,\
 V \in L^{\gamma + \frac{d}{2}}(C), \ \int_C V(x)_-^{\gamma + \frac{d}{2}}\rd x=I^{\gamma + \frac{d}{2}} \bigg\}
\end{equation}
where $\eps^V_n(\xi)$ is the $n$th eigenvalue of the operator $H_\xi=|-i\nabla+\xi|^2+V$ with periodic boundary conditions on $\partial C$ and quasimomentum $\xi$. The Lieb-Thirring constant equals
$$ L_{\gamma, d} := \sup_{I > 0,\ K\in\N,\ \cL} \frac{L_{\gamma,d,\cL}(K,I)}{I^{\gamma+\frac{d}2}}.$$
After scaling, varying $I$ is the same as changing the period of the potential. 

We solve~\eqref{eq:def:M} with an iterative fixed point-type algorithm. At each iteration $n$, we compute the $K$ first eigenvectors $u_{j, \xi}^{(n)}$ of $-\Delta + V^{(n)}$ with quasimomentum $\xi$, and set
\begin{equation}
     \rho^{(n)}(x) := \frac1{|B|}\sum_{j=1}^K\int_B \eps^{V^{(n)}}\!(\xi)_-^{\gamma - 1} | u^{(n)}_{j, \xi} (x)|^2\, \rd \xi, \qquad
    V^{(n+1)}(x) = - a_n\, \rho^{(n)}(x)^{\frac1{\gamma + d/2 - 1}},
\label{eq:SCF_num}
\end{equation}
with the constant $a_n>0$ chosen so that $\|V^{(n+1)}\|_{L^{\gamma + d/2}(C)} = I$. The corresponding objective function in~\eqref{eq:def:M} can be seen to increase with the iterations. We stop the algorithm when $\|V^{(n+1)} - V^{(n)}\|_{L^{\gamma + d/2}(C)}$ is smaller than a prescribed small parameter. For the initial potential $V^{(0)}$ we use a periodic arrangement of Gaussians. 

We represent a potential $V$ by its values on a $(N_C)^d$ regular grid in $C$. Since the obtained potentials seem smooth, the Riemann sum converges fast to $\int_C V_-^{\gamma + d/2}$ as $N_C$ gets large. The Brillouin zone integration is computed on a $(N_B)^d$ regular grid in $B$. When the operator $-\Delta+V$ has a gap above its $K$th band (which was always the case in our computations), the Brillouin zone integration converges exponentially fast in $N_B$. 

\medskip

\noindent\textbf{Results in one dimension for $K=1$ Bloch band.} Using Theorem~\ref{thm:1D} and the fact that $k\in(0,1)\mapsto \| \widetilde{V_k} \|_{L^2(0,1)}$ is increasing from $\pi^2$ to $+\infty$, where $\widetilde{V_k}$ is the $1$-periodic rescaled version of $V_k$ of Theorem~\ref{thm:1D}, one can prove that for $\gamma=1/2$, the problem $L_{\frac32, 1,\Z}(1,I)$ admits as maximizer the constant potential $V = - I$ for $I\in(0,\pi^2]$, and some $\widetilde{V_k}$ for $I\in[\pi^2,\infty)$. For $I < \pi^2$, the corresponding Hamiltonian is gapless, while for $I > \pi^2$, the first band is isolated from the rest by a gap of size $k^2$.

In Figure~\ref{fig:d=1}, we provide numerical results for $L_{\gamma, 1,\Z}(1,I)$ with $\gamma  \in [0.6, 2]$ and for different values of $I > 0$. Each $I>0$ seems to give rise to a branch of periodic optimizers. When $I<\pi^2$ (that is, $I\in\{2,\dots,8\}$ in the picture), the branch coincides with the semiclassical constant in a neighborhood of $\gamma=3/2$. When $I>\pi^2$ the branch passes through the (rescaled) solution $V_k$ at $\gamma=3/2$ and the corresponding potentials are never constant. All the curves cross at $\gamma = 3/2$, as expected from Theorem~\ref{thm:1D}. For $\gamma < 3/2$, the curves are all below the one-bound-state constant $L_{\gamma, 1}^{(1)}$, while for $\gamma > 3/2$, they are all below the semiclassical constant $L_{\gamma, 1}^{\rm sc}$. This is in agreement with the Lieb-Thirring conjecture in one dimension. 

\begin{figure}[t]
\centering
    \includegraphics[width=6cm]{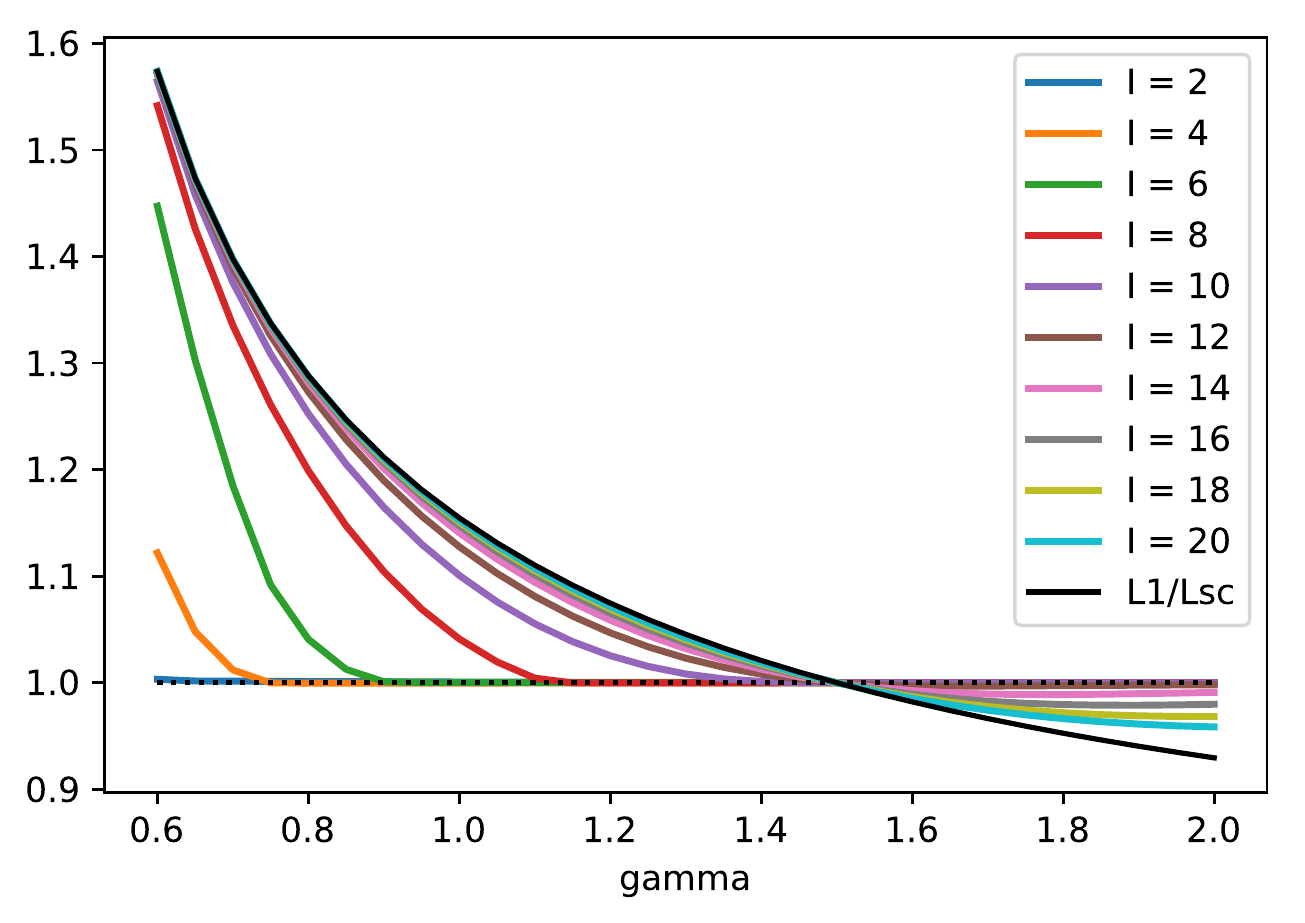} \includegraphics[width=6cm]{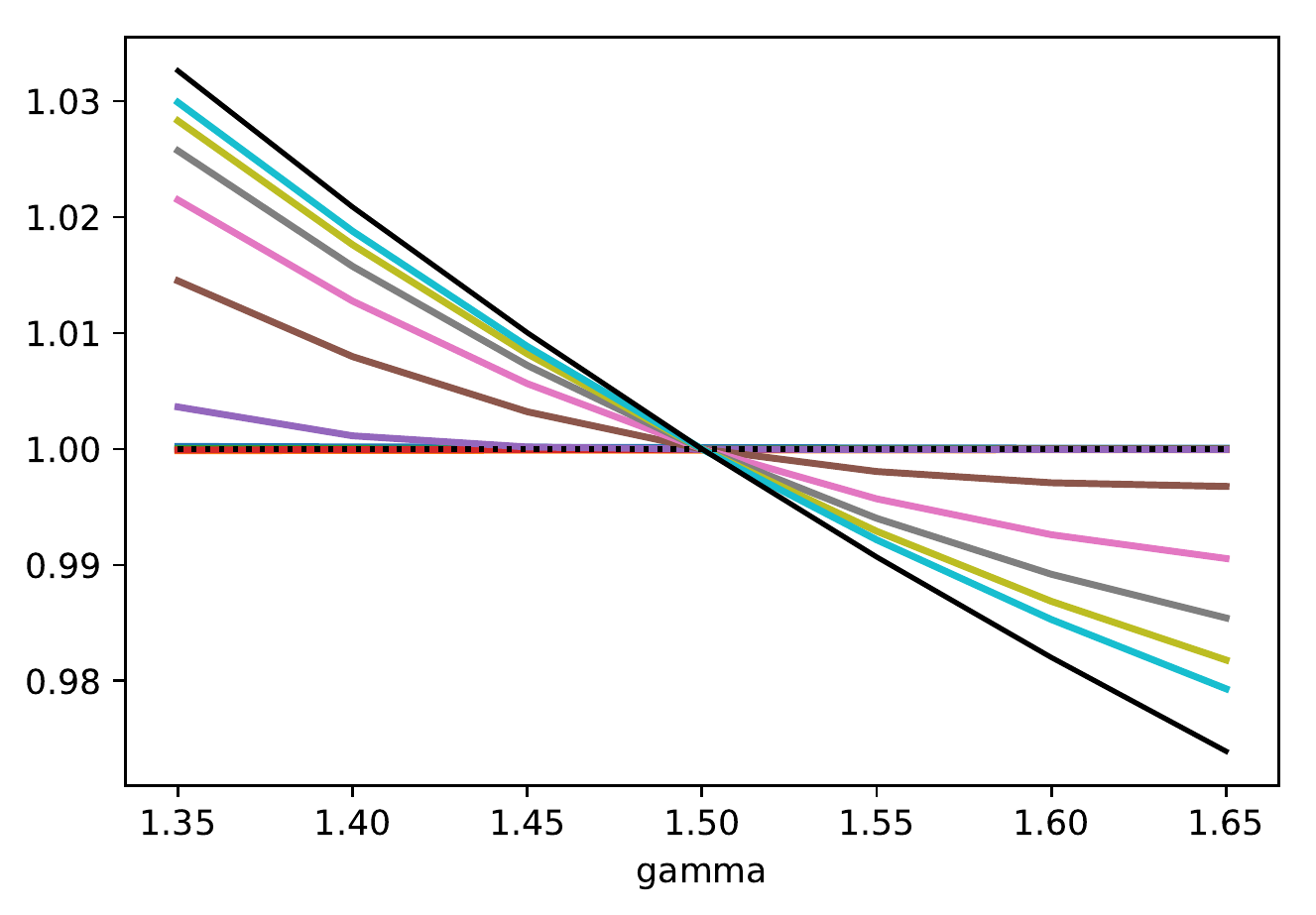}
    \caption{$L_{\gamma,1,\Z}(1,I) / L_{\gamma, 1}^{\rm sc}$ for different values of $I>0$. The black curve is $L_{\gamma, 1}^{(1)} / L_{\gamma, 1}^{\rm sc}$.}
    \label{fig:d=1}
\end{figure}

\medskip

\noindent\textbf{Results in two dimensions.} In dimension $d=2$, we recall that the curves $L_{\gamma, 2}^{(1)}$ and $L_{\gamma, 2}^{\rm sc}$ cross at $\gamma_{1 \cap {\rm sc}}(2) \simeq 1.165378$. In~\cite{FraGonLew-20_ppt}, we proved that the critical exponent at which the semiclassical constant becomes optimal satisfies $\gamma_{\rm sc}(2) > \gamma_{1 \cap {\rm sc}}(2)$. We will here provide numerical evidence of the strict inequality, but also that $\gamma_{\rm sc}(2)$ could be quite close to $\gamma_{1 \cap {\rm sc}}(2)$. As we will see, a very high precision is needed to be able to compare it with the semi-classical and one-bound state constants. We took $(N_C,N_B)=(40,30)$ and the computation of the optimal potential for one value of $I$ and $\gamma$ took approximately one hour on one processor. Several points were handled simultaneously using parallel computing. 

In Figure~\ref{fig:d=2_kappa11654}, we display the curves $I\mapsto L_{\gamma, 2,\cL}(K,I)$ for $\gamma\in\{1.165300,1.165400\}$, that is, slightly below and above $\gamma_{1 \cap {\rm sc}}(2)$. We considered three different lattices $\cL$: triangular ($K=1$), square ($K=1$) and hexagonal ($K=2$). The scale in the $y$-axis is very fine and all quantities are computed to the order $10^{-7}$. The curves are computed by solving~\eqref{eq:def:M} on a grid, while the horizontal black line $L^{(1)}_{\gamma,d}/L^{\rm sc}_{\gamma,d}$ is computed by finding the positive solution of the nonlinear Schr\"odinger equation using Runge-Kutta methods. The fact that the curves get very close when $I$ increases strongly suggests that both numerical codes are valid with very high accuracy in this region. 

At  $\gamma=1.165300<\gamma_{1 \cap {\rm sc}}(2)$, for each of the three lattices we can find a value of $I$ so that the corresponding \textbf{periodic potential beats both the semiclassical and one-bound-state constants}. The triangular lattice provides the largest constant. When we increase $\gamma$, the three curves go down and end up touching the semiclassical constant slightly after $\gamma_{1\cap\rm sc}(2)$. The touching points are provided in Table~\ref{table:critical_gamma}. That the critical exponents are so close to $\gamma_{1\cap \rm sc}(2)$ could be a consequence of the exponentially small attraction between the particles~\cite{FraGonLew-20_ppt}. For the second curve in Figure~\ref{fig:d=2_kappa11654} we have $\gamma=1.165400>\gamma_{1 \cap {\rm sc}}(2)$, and only the triangular lattice is above the semiclassical constant. In Figure~\ref{fig:d=2_density} we display the periodic potentials at $\gamma=1.165400$. We note that in all cases the obtained optimal potentials had exactly $K$ negative bands.

If we double the period, which corresponds to taking $4K$ bands, we would obtain similar curves as in Figure~\ref{fig:d=2_kappa11654} with a maximum located at $I$ equal to about 4 times the values in Table~\ref{table:critical_gamma}. This maximum must be at least as high as the one we got for $K$ bands. When we allow more and more bands, we in fact know from the proof of Theorem~\ref{thm:LT_periodic} that this local maximum converges to the optimal Lieb-Thirring constant, whatever lattice $\cL$ we start with. In other words, the function $I\mapsto \max_{K\geq1}L_{\gamma,d,\cL}(K,I)$ in~\eqref{eq:def:M} is oscillating and converging to the best Lieb-Thirring constant $L_{\gamma,d}$ in the limit $I\to\ii$. In Figure~\ref{fig:d=2_kappa11654} we only display the first bump of this function. Our Conjecture~\ref{conjecture} would mean that for a certain lattice $\cL$ these maxima are all the same. It is unfortunately hard to simulate a too large number of Bloch bands since more discretization points are then needed.

To conclude, \textbf{above $\gamma_{1\cap\rm sc}(2)$ we have found a periodic potential which beats the semiclassical constant and have thus shown that $\gamma_{\rm sc}(2)\geq 1.165417$}. Even \textbf{slightly below $\gamma_{1\cap\rm sc}(2)$, periodic potentials can do better than the one-bound state constant}. This shows that periodic potentials are important for the Lieb-Thirring inequality and gives evidence to our conjecture that they are optimizers. 

\begin{figure}[!t]
    \centering
    
    \hspace{-0.1cm}\includegraphics[width=6.6cm]{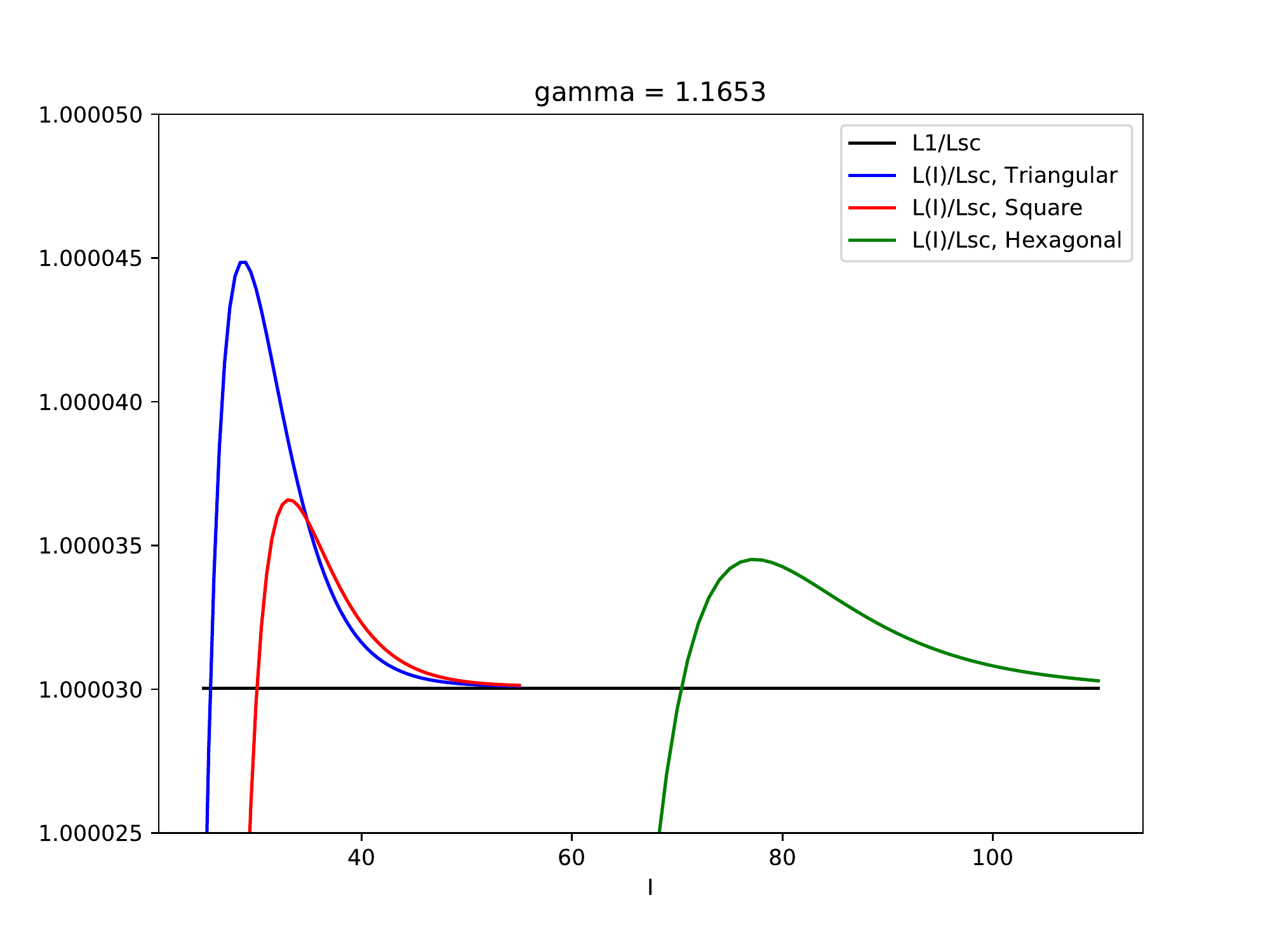}\includegraphics[width=6.6cm]{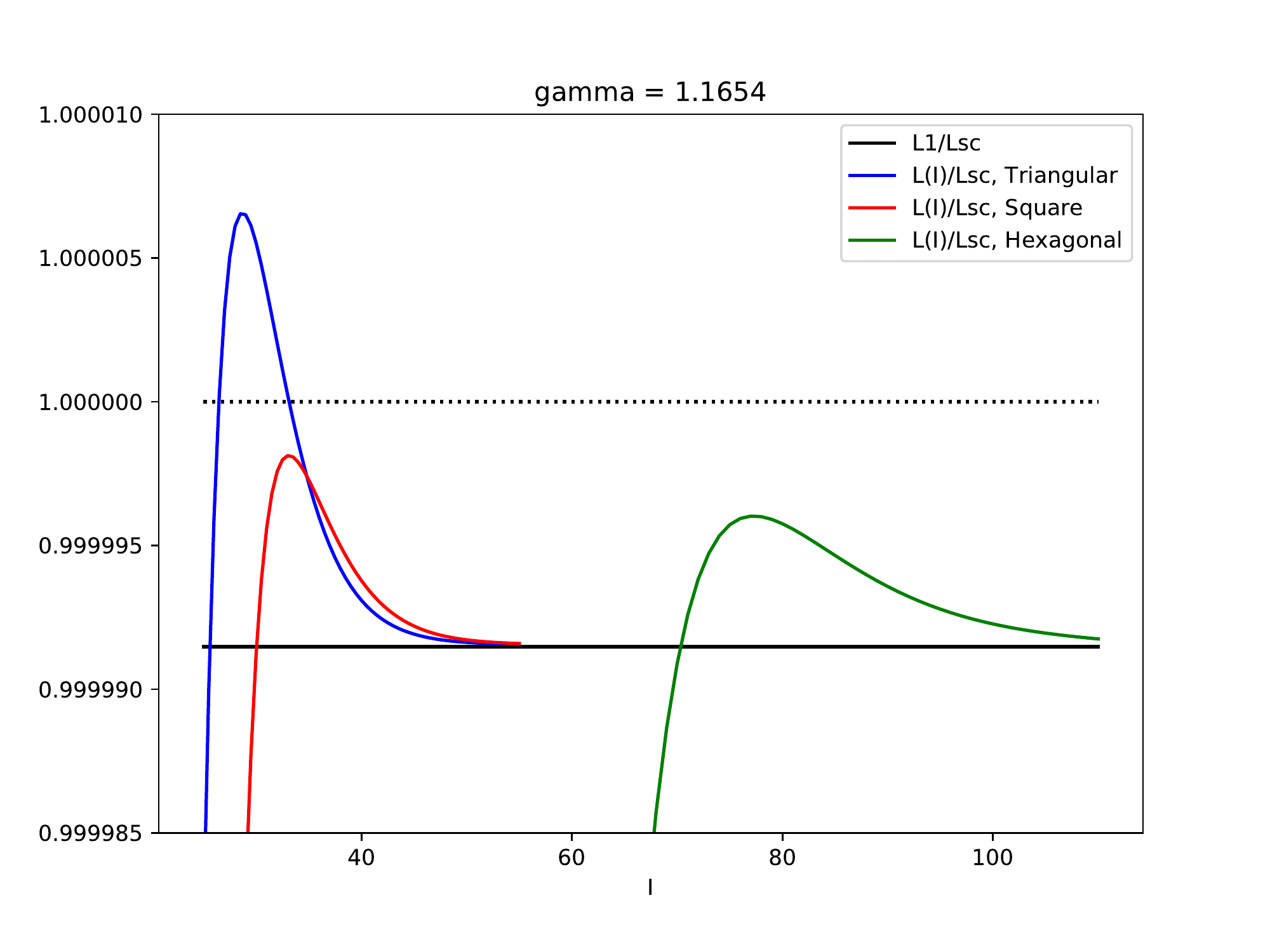}
    
    \caption{Functions $I\mapsto L_{\gamma,2,\cL}(N,I) / L_{\gamma, 2}^{\rm sc}$ for $\gamma=1.165300$ (left) and $\gamma=1.165400$ (right). The black horizontal line is the constant $L_{\gamma, 2}^{(1)}/L_{\gamma, 2}^{\rm sc}$. Note that the dotted line on the right is not the curve obtained for the constant potential $V\equiv-I$. The latter lies much further down for these values of $I$ since we only retain $K$ bands and hence do not fill the whole negative spectrum of $-\Delta-I$. 
    \label{fig:d=2_kappa11654}}
    
    \bigskip
    
    \includegraphics[height=5cm]{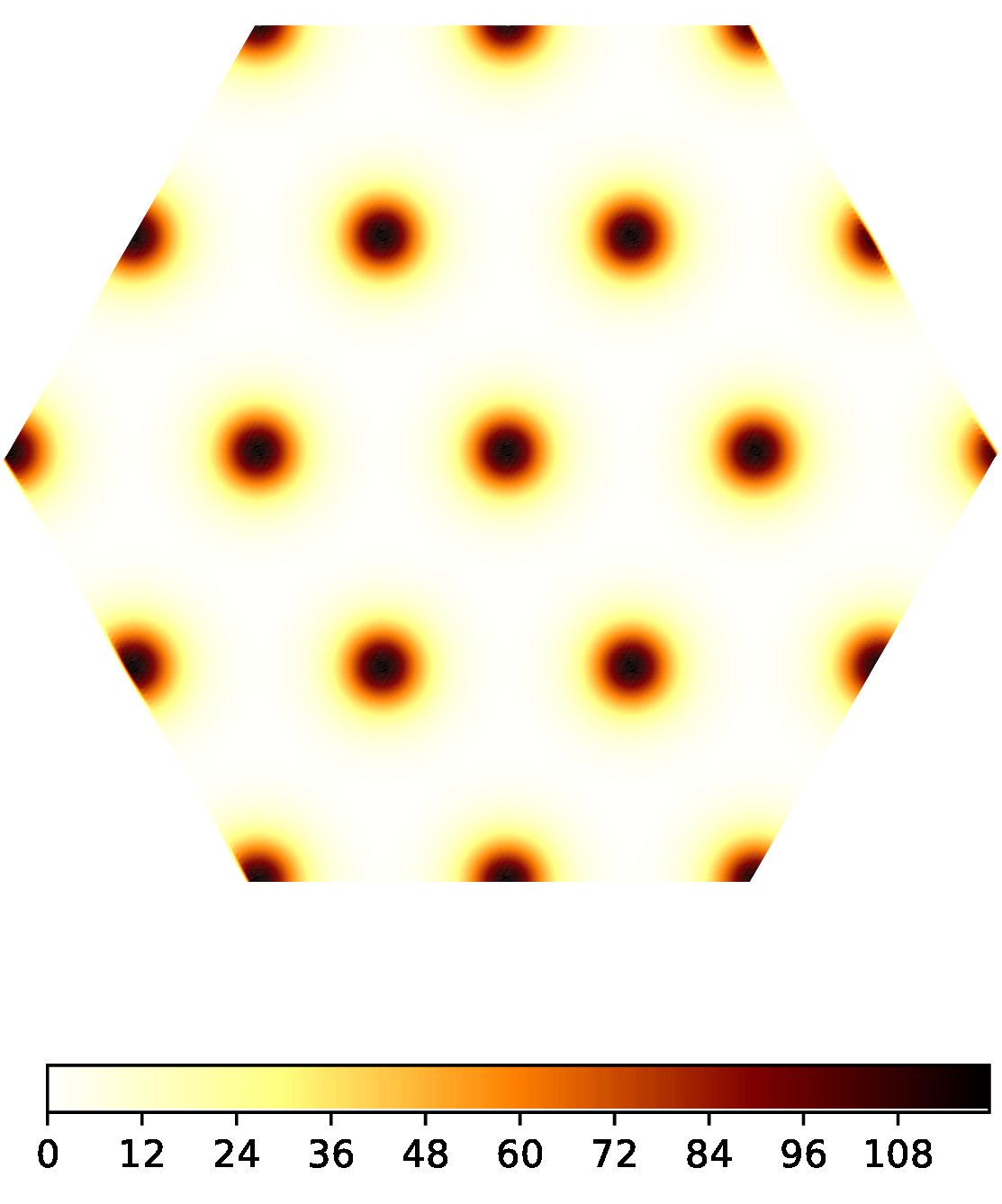}\includegraphics[height=5cm]{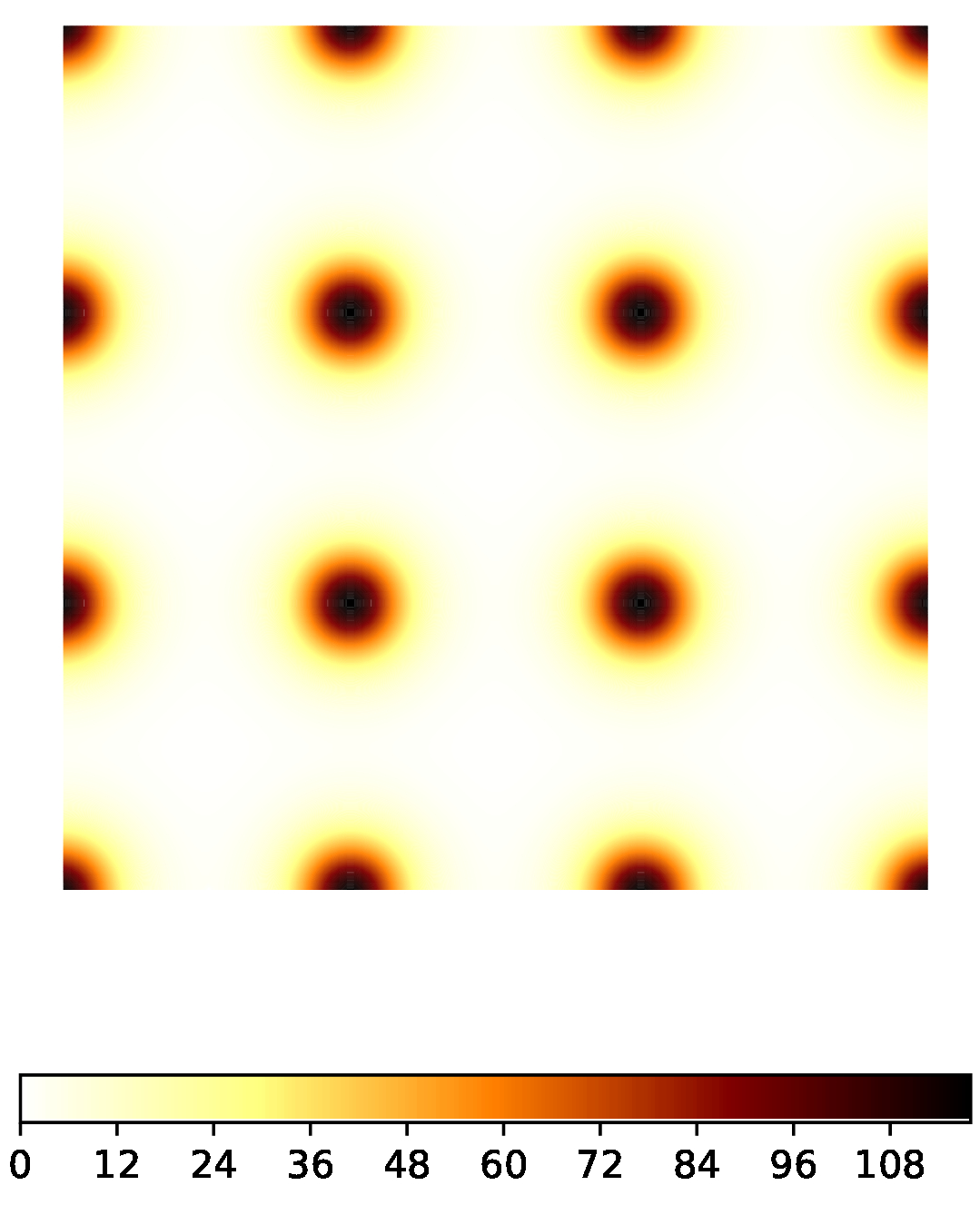}\includegraphics[height=5cm]{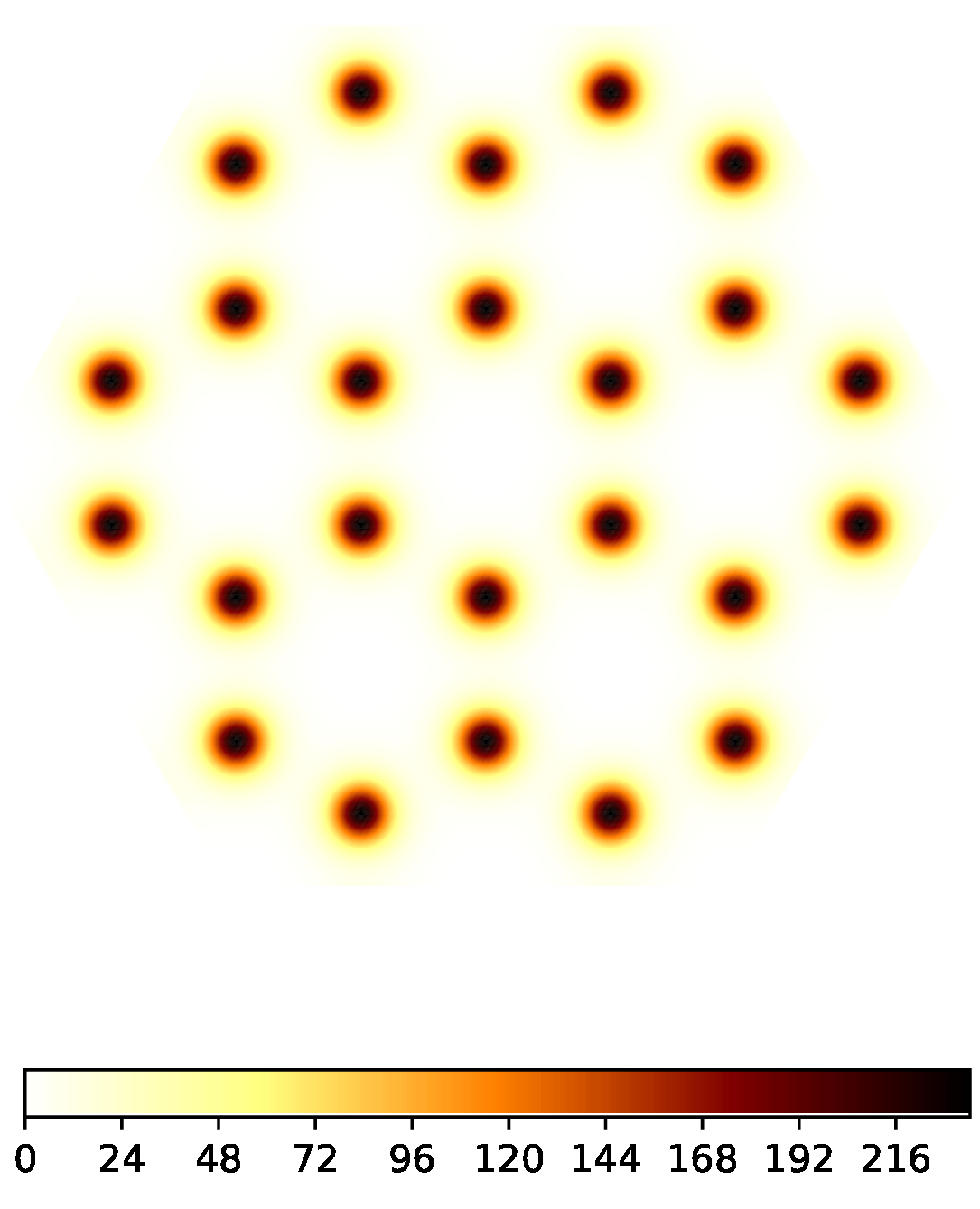}
    
    \caption{Absolute value of the optimal potential at $\gamma=1.165400$ for the triangular (left), square (center) and hexagonal lattices (right), with the corresponding best $I$.\label{fig:d=2_density}}
    \end{figure}

        \begin{table}[!t]
    \centering
        \begin{tabular}{|c|c|c|c| c |}
            \hline
            & Triangular & Square & Hexagonal & $L^{(1)}_{\gamma,2}$ \\
            \hline
            Critical $\gamma$ & 1.165417 & 1.165395 & 1.165390 & 1.165378 \\
            \hline
            Corresponding $I$ & 28.7 & 33.1 & 77.2 & -- \\
            \hline
        \end{tabular}
        \caption{Approximate critical values of $\gamma$ at which $\sup_I L_{\gamma,2,\cL}(K,I) = L_{\gamma, 2}^{\rm sc}$, for different lattices $\cL$ (with the corresponding value of $I$), and for $L_{\gamma, 2}^{(1)}$.\label{table:critical_gamma}}
\end{table}

\appendix
\section{A minimization problem with entropy}\label{app:entropy}

Consider the problem of minimizing the mean-field free energy of a (bosonic) quantum system described by its one-particle density matrix $\Gamma$ on $L^2(\R^d)$, with a local nonlinear attraction and a Tsallis-type entropy
\begin{equation}
 F_{p,d}(T)=\inf_{\Gamma=\Gamma^*\geq0}\left\{\tr(-\Delta)\Gamma-\frac1p\int_{\R^d}\Gamma(x,x)^p\,\rd x+T\tr(\Gamma^q)\right\}
 \label{eq:F_entropy}
\end{equation}
where $T$ plays the role of a temperature. We take $1\le p<1+\frac2d$ to be able to rely on Lieb-Thirring inequalities and properly set-up the problem. At $q = (2p+d-dp)/(2+d-dp)$ the problem is scaling-invariant and it follows that $F_{p,q}(T)=0$ for $T\geq T_c(p,d)$ with no minimizer for $T> T_c(p,d)$ and $F_{p,q}(T)=-\ii$ for $T<T_c(p,d)$. The critical temperature $T_c(p,d)$  is positive and finite. After scaling $\Gamma$, it can be computed in terms of the best constant in the inequality
\begin{equation}
 K_{p,d} \norm{\rho_\Gamma}_{L^p(\R^d)}^{\frac{2p}{d(p-1)}} \le  \big(\tr(\Gamma^q)\big)^{\frac{p(2-d)+d}{dq(p-1)}}\tr(-\Delta)\Gamma.
 \label{eq:LT_Schatten}
\end{equation}
The inequality stays valid for $p=1+2/d$ and $q=+\ii$ with $ \big(\tr(\Gamma^q)\big)^{1/q}$ replaced by the operator norm $\|\Gamma\|$ but in~\eqref{eq:F_entropy} this becomes a constraint $\|\Gamma\|\leq1$. For simplicity, we assume $p<1+2/d$. In~\cite{LioPau-93,FraGonLew-20_ppt}, the inequality~\eqref{eq:LT_Schatten} was shown to be dual to the Lieb-Thirring inequality~\eqref{eq:LT_V}, when $p=(\gamma+d/2)'$ and $q=\gamma'$, so that $K_{p,d}$ can be expressed in terms of $L_{\gamma,d}$. The periodic equivalent of Theorem~\ref{thm:LT_periodic} is

\begin{thm}[Periodic dual Lieb-Thirring inequality]
Let $d\geq1$ and $1<p<1+2/d$. For every bounded periodic operator $\Gamma=\Gamma^*\geq0$ we have 
\begin{equation}
 K_{p,d} \left(\fint \rho_\Gamma^p\right)^{\frac{2}{d(p-1)}} \le  \big(\VTr(\Gamma^q)\big)^{\frac{p(2-d)+d}{dq(p-1)}}\VTr(-\Delta)\Gamma,
 \label{eq:LT_Schatten_periodic}
\end{equation}
with the same optimal constant $ K_{p,d}$ as in~\eqref{eq:LT_Schatten}.
\end{thm}

Any periodic optimizer for the inequality~\eqref{eq:LT_Schatten_periodic}, if it exists, solves the nonlinear equation
$$\Gamma=\left(-\Delta-a\rho_\Gamma^{\frac1{\gamma+d/2-1}}\right)_-^{\gamma-1}$$
with $p=(\gamma+d/2-1)^{-1}$, $q=(\gamma-1)^{-1}$ and some appropriate constant $a>0$. This is the equation used in our fixed point algorithm~\eqref{eq:SCF_num}.

\section{Computation of the density of states of $V_k$}\label{app:details}

Our goal in this appendix is to prove the formula~\eqref{eq:nE} for the density of states of $V_k$. For the convenience of the reader, we will provide here all the necessary tools from the theory of elliptic functions. There is a well developed theory of periodic Schr\"odinger operators in 1D with finitely many gaps, which is closely connected to integrable systems. The simplest case is for one-gap potential like $V_k$, where everything relies on elliptic functions. Recall that there are two theories of elliptic functions due to Weierstrass and Jacobi, respectively. We have stated our main theorem in terms of Jacobi elliptic functions, but it will be more convenient to compute the density of states in the set-up of Weierstrass elliptic functions. We will do this in Section~\ref{sec:Schrodinger_1D}. In the first section we quickly review the theory of elliptic functions.

\subsection{Weierstrass theory of elliptic functions}


The whole discussion depends on two parameters $\omega_1,\omega_2\in\C\setminus\{0\}$ with $\frac{\omega_2}{\omega_1}\not\in\R$.\footnote{There are two conflicting notational conventions concerning the periods of elliptic functions. We follow here~\cite[Chap.~7]{Ahlfors} and denote by $\omega_1$ and $\omega_2$ the periods. Sometimes $\omega_1,\omega_2$ denote instead the \emph{half-periods}.} For our need, we focus on the case $\omega_1 \in \R$ (at the end, $\omega_1 = 2 K(k)$ is the period), and $\omega_2 \in i \R$. This choice simplifies some proofs, and we refer to~\cite{Ahlfors, Simon2A} for a general discussion. Associated with these two numbers is a Weierstrass $\wp$-function
$$
\wp(z) := \frac{1}{z^2} + \sum_{n\in\Z^2\setminus\{0\}} \left( \frac{1}{(z-n_1\omega_1-n_2\omega_2)^2} - \frac{1}{(n_1\omega_1+n_2\omega_2)^2} \right),
\qquad z\in\C \,.
$$
This function is meromorphic with double poles on the grid $\{n_1\omega_1+n_2\omega_2:\ n\in\Z^2 \}$, and periodic with periods $\omega_1$ and $\omega_2$. We always have $\wp(-z) = \wp(z)$, and since $\omega_1 \in \R$ and $\omega_2 \in i \R$, we also have
\begin{equation} \label{eq:wp_real_on_grid}
    \forall x \in \R, \quad \wp\left(x - n_2 \frac{\omega_2}{2}\right) \in \R \quad \text{and} \quad \wp\left(i x - n_1 \frac{\omega_1}{2}\right) \in \R.
\end{equation}
On the positively oriented rectangle 
\begin{equation} \label{eq:loop}
    \mathcal{C} := \left[0, \frac{\omega_1}{2}\right] \cup \left[\frac{\omega_1}{2}, \frac{\omega_1 + \omega_2}{2}\right] \cap \left[ \frac{\omega_1 + \omega_2}{2} , \frac{\omega_2}{2}\right] \cup \left[\frac{\omega_2}{2}, 0\right],
\end{equation}
the function $\wp$ is real-valued, continuous except at $0$, satisfies $\wp(x) \to \infty$ and $\wp(i x) \to - \infty$ as $\R \ni x \to 0$. Since $\wp$ has order $2$, for any $w \in \C$, the equation $\wp(z) = w$ has exactly two solutions in a fundamental cell. This implies first that $\wp$ is real-valued only on the grid~\eqref{eq:wp_real_on_grid}, and that $\wp$ is strictly decreasing along the loop. In particular,
\[
    e_1 > e_2 > e_3, \quad \text{with} \quad
    e_1 := \wp\left(\frac{\omega_1}{2}\right), \quad  e_2 := \wp\left(\frac{\omega_1 + \omega_2}{2}\right), \quad e_3 := \wp\left(\frac{\omega_2}{2}\right).
\]
By evenness and periodicity of $\wp$ we also have $ \wp'\left(\frac{\omega_1}{2}\right) =  \wp'\left(\frac{\omega_1 + \omega_2}{2}\right) =  \wp'\left(\frac{\omega_2}{2}\right) = 0$. Since $\wp'$ is of order $3$, these are the only roots of $\wp'$ in the fundamental cell. We set
$$
g_2 = 60 \sum_{n\in\Z^2\setminus\{0\}} \frac{1}{(n_1\omega_1+n_2\omega_2)^4}
\qquad\text{and}\qquad
g_3 = 140 \sum_{n\in\Z^2\setminus\{0\}} \frac{1}{(n_1\omega_1+n_2\omega_2)^6} \,. 
$$
Comparing the expansions of $\wp'$ and $\wp$ near the poles, we obtain
\begin{equation}
\label{eq:wpdiffeq}
\wp'(z)^2 = 4\wp(z)^3 - g_2\wp(z) - g_3 = 4 ( \wp(z)  - e_1) (\wp(z) - e_2)(\wp(z) - e_3),
\end{equation}
where we used in the last equality that $\wp'$ vanishes at $\frac{\omega_1}{2}$, $\frac{\omega_1 + \omega_2}{2}$ and $\frac{\omega_2}{2}$. Together with the fact that $\wp$ is decreasing along the loop, this gives by separation of variables that
\begin{equation} \label{eq:x-x0}
a-a_0 =  i \int_{\wp(a_0)}^{\wp(a)} \frac{\rd w}{\sqrt{4(e_1-w)(e_2-w)(e_3-w)}}
\quad\text{for all}\ a,a_0\in (0, \tfrac12 \omega_2].
\end{equation}
We have similar formulae on the other parts of the loop $\mathcal{C}$, that we omit for brevity. \\

Before we turn to the link between the Weierstrass $\wp$ function and the Schr\"odinger equation, we need two more functions. The first one is the Weierstrass zeta function
$$
\zeta(z) := \frac 1z + \sum_{n\in\Z^2\setminus\{0\}} \left(\frac{1}{z-n_1\omega_1-n_2\omega_2} + \frac{1}{n_1\omega_1+n_2\omega_2} + \frac{z}{(n_1\omega_1+n_2\omega_2)^2} \right),
\quad z\in\C .
$$
It has which has simple poles at $\{n_1\omega_1+n_2\omega_2:\ n\in\Z^2 \}$ and satisfies
\begin{eqnarray}
\label{eq:zetadiff}
\zeta'(z)=-\wp(z) \,.
\end{eqnarray}
We have $\zeta(-z) = -\zeta(z)$, $\zeta(z+\omega_1) = \zeta(z)+\eta_1$ and $\zeta(z+\omega_2) = \zeta(z)+\eta_2$, for
\begin{equation}
\label{eq:eta}
\eta_1 := 2\zeta\left(\frac{\omega_1}2\right)
\qquad\text{and}\qquad
\eta_2 := 2 \zeta\left(\frac{\omega_2}2\right).
\end{equation}
In our case where $\omega_1 \in \R$ and $\omega_2 \in i \R$, we get from~\eqref{eq:zetadiff} and the fact that $\wp$ is real-valued on the loop $\mathcal{C}$, that $\zeta \left( x \right) \in \R$ while $ \zeta \left( i x \right) \in i \R$ for all $x \in \R$. This shows for instance that $\eta_1 \in \R$ and $\eta_2 \in i \R$. Moreover, one can show~\cite{Ahlfors} that
\begin{eqnarray}
\label{eq:legendre}
\eta_1 \omega_2 - \eta_2\omega_1 = 2\pi i.
\end{eqnarray}
As in~\eqref{eq:x-x0}, using~\eqref{eq:zetadiff} we see that
\begin{equation} \label{eq:zetax-zetax0}
  \zeta(a)-\zeta(a_0) = - i \int_{\wp(a_0)}^{\wp(a)} \frac{w\,\rd w}{\sqrt{4(e_1 - w)( e_2 - w)(e_3 - w)}}, \quad
  \text{for all}\ a,a_0\in (0, \tfrac12 \omega_2].
\end{equation}
The last function to be introduced is
$$
\sigma(z) := z \prod_{n\in\Z^2\setminus\{0\}} \left( 1 - \frac{z}{n_1\omega_1+n_2\omega_2} \right) e^{\frac{z}{n_1\omega_1+n_2\omega_2} + \frac12\,\frac{z^2}{(n_1\omega_1+n_2\omega_2)^2}} \,,
\qquad z\in\C \,.
$$
This is an odd entire function which satisfies
\begin{equation}
\label{eq:sigmadiff}
\frac{\sigma'(z)}{\sigma(z)} = \zeta(z) \,,\quad 
\sigma(z+\omega_1) = -\sigma(z) e^{\eta_1(z+\omega_1/2)},
\quad
\sigma(z+\omega_2) = -\sigma(z) e^{\eta_2(z+\omega_2/2)} \,.
\end{equation}

\medskip

\subsection{The Schr\"odinger equation}\label{sec:Schrodinger_1D}

We let $\omega_1\in (0,\infty)$ and $\omega_2\in i(0,\infty)$ as before, and denote by $\wp$ the corresponding Weierstrass function. We consider the potential
$$
W(x) := 2\,\wp\left(x+\tfrac12\omega_2\right) \,,
\qquad x\in\R \,.
$$
It has (real) period $\omega_1$, is real-valued and real analytic (since the line $\R+\omega_2/2$ does not hit the poles of $\wp$). Moreover, it is even and symmetric about $x=\frac12\omega_1$. It is strictly increasing on $[0, \tfrac{1}{2}\omega_1]$ with $W(0) = 2e_3$ and $W(\tfrac12\omega_1)=2e_2$. We study the (periodic) Schr\"odinger equation in the form
\begin{eqnarray}
\label{eq:soeq}
\begin{cases}
-\psi'' + W\psi = E \psi
\qquad\text{on}\ \R \,,\\
E= - \wp(a).
\end{cases}
\end{eqnarray}
Here, $a \in \C$ is any parameter. Soon we will require $E \in \R$, in which case $a$ must belong to the loop $\mathcal{C}$ defined in~\eqref{eq:loop}. As $a$ runs positively through this loop, $E$ goes from $-\infty$ to $\infty$.
Let
$$
\psi_\pm(x) :=  \frac{\sigma(x+\frac12\omega_2\pm a)}{\sigma(x+\frac12\omega_2)} e^{\mp (x+\frac12\omega_2)\zeta(a)}\,.
$$

\begin{lem}
	The functions $\psi_\pm$ solve \eqref{eq:soeq}. If $a$ belongs to the fundamental domain and $a\not\in\left\{0,\frac12\omega_1,\frac12\omega_2,\frac12(\omega_1+\omega_2)\right\}$, then $\psi_+$ and $\psi_-$ are linearly independent.
\end{lem}

We will see that the excluded values of $a$ give the boundary of the spectrum.

\begin{proof}
	We have, using \eqref{eq:sigmadiff},
	\begin{align*}
	\psi_\pm'(x) & = e^{\mp (x+\frac12\omega_2)\zeta(a)} \bigg( \mp \zeta(a) \frac{\sigma(x+\frac12\omega_2\pm a)}{\sigma(x+\frac12\omega_2)} + \frac{\sigma'(x+\frac12\omega_2\pm a)}{\sigma(x+\frac12\omega_2)}\\
	&\qquad - \frac{\sigma(x+\frac12\omega_2\pm a)\sigma'(x+\frac12\omega_2)}{\sigma(x+\frac12\omega_2)^2} \bigg) \\
	& = \psi_\pm(x) \left( \mp \zeta(a) + \zeta(x+\tfrac12\omega_2\pm a) - \zeta(x+\tfrac12\omega_2) \right)
	\end{align*}
	and this gives 
	\begin{align*}
	\psi_\pm''(x) 
	& = \psi_\pm(x) \bigg\{ \left( \mp \zeta(a) + \zeta(x+\tfrac12\omega_2\pm a) - \zeta(x+\tfrac12\omega_2) \right)^2\\ 
	&\qquad + \left( \zeta'(x+\tfrac12\omega_2\pm a) - \zeta'(x+\tfrac12\omega_2) \right) \bigg\}.
	\end{align*}
	Thus, to prove the lemma, we need to show that the term in the curly bracket equals $2\wp(x+\tfrac12 \omega_2) +\wp(a)$. According to \eqref{eq:zetadiff}, this is the same as
	$$
	\left( \mp \zeta(a) + \zeta(z\pm a) - \zeta(z) \right)^2 
	+ 	\zeta'(z\pm a) + \zeta'(z) + \zeta'(a) = 0
	$$
	for all $z=x+\frac12\omega_2$. Using the oddness of $\zeta$ we rewrite the quantity under the square as
	$
	\mp \zeta(a) + \zeta(x\pm a) - \zeta(x) = \zeta(\mp a) + \zeta(x\pm a) + \zeta(-x)
	$
	and note that the three numbers involved in the right side satisfy
	$(\mp a) + (z\pm a) + (-z) = 0$. Therefore, by~\cite[Exercise 5, Sec.~10.45]{Simon2A} and the evenness of $\zeta$, we have
	\begin{align*}
	\left( \zeta(\mp a) + \zeta(z\pm a) + \zeta(-z) \right)^2 & = - \left( \zeta'(\mp a) + \zeta'(z\pm a) + \zeta'(-z) \right)\\
	&=-\left(\zeta'(a) + \zeta'(z\pm a) + \zeta'(z)\right),
	\end{align*}
	which proves the claimed identity. In order to show that the two functions are linearly independent, we compute their Wronskian, using the above formulas for $\psi_\pm'$:
	\begin{multline*}
	\psi_+\psi_-' - \psi_-\psi_+' 
	\\ = \left( 2 \zeta(a) +\zeta(x+\tfrac12\omega_2-a) - \zeta(x+\tfrac12\omega_2+a) \right)\frac{\sigma(x+\frac12\omega_2+a)\sigma(x+\frac12\omega_2- a)}{\sigma(x+\frac12\omega_2)^2}
	\,.
	\end{multline*}
	According to~\cite[Eq.~(10.4.90)]{Simon2A}  we have
	$$
	\frac{\sigma(x+\frac12\omega_2+a)\sigma(x+\frac12\omega_2- a)}{\sigma(x+\frac12\omega_2)^2} = \left(\wp(a)-\wp(x+\tfrac12\omega_2)\right)\sigma(a)^2
	$$
	and according to oddness of $\zeta$ and~\cite[Eq.~(10.4.91)]{Simon2A}, we have
	\begin{multline*}
	2 \zeta(a) +\zeta(x+\tfrac12\omega_2-a) - \zeta(x+\tfrac12\omega_2+a)\\
	 = 2 \zeta(a) - \zeta(a - (x+\tfrac12\omega_2)) - \zeta(a+(\tfrac12\omega_2)) 
	 = - \frac{\wp'(a)}{\wp(a)-\wp(x+\tfrac12\omega_2)} \,.
	\end{multline*}
	Thus we find $\psi_+\psi_-' - \psi_-\psi_+' = - \sigma(a)^2 \wp'(a)$. Since $\sigma$ does not vanish in the fundamental domain except at the origin (this follows from the product formula) and since $\wp'$ vanishes exactly at the excluded points, we see that the Wronskian is nonzero for $a$ in the fundamental domain and different from the excluded values.
\end{proof}

\paragraph{The dispersion relation.} 
From \eqref{eq:sigmadiff} we have $\psi_\pm(x+\omega_1) = e^{\mp\omega_1\zeta(a)\pm \eta_1 a} \psi_\pm(x)$.  Setting
\begin{equation} \label{eq:xi(a)}
    \xi_\pm := \mp i (\zeta(a)-\frac{\eta_1}{\omega_1}a), \quad \text{modulo $2 \pi$,}
\end{equation}
we have $\psi_\pm(x+\omega_1) := e^{-i\omega_1 \xi_\pm} \psi_\pm(x)$, so $\psi_\pm$ are Bloch--Floquet solutions for the energy $E = \wp(a)$ whenever $\xi_\pm \in \R$. Since we want $E \in \R$, we need to take $a \in \mathcal{C}$. Together with~\eqref{eq:zetadiff}, we can see that $\xi_\pm \in \R$ if and only if $a$ belongs to the vertical segments $a \in (0, \frac12 \omega_2] \cup [\tfrac12 \omega_1, \tfrac12 (\omega_1+ \omega_2)]$. This already proves that $-\partial_x^2 + W$ has spectrum $[-e_1, -e_2] \cup [-e_3, \infty)$, and in particular has a unique gap. 

\medskip

Formula~\eqref{eq:xi(a)} gives an implicit parametrization of the dispersion curves. We will now use the integral representations to rewrite $\xi_\pm$ in terms of $E$ instead of $a$. First consider $a \in (0, \tfrac12 \omega_2]$, that is $E \in (\infty, -e_3]$. We set $a_0 := \tfrac12 \omega_2$, and get from~\eqref{eq:x-x0}-\eqref{eq:zetax-zetax0} that
\[
    \zeta(a) - \frac{\eta_1}{\omega_1} a  = \left( \zeta(a_0) - \frac{\eta_1}{\omega_1}a_0\right) - i \int_{\wp(a_0)}^{\wp(a)} \dfrac{(w + \tfrac{\eta_1}{\omega_1}) \rd w}{\sqrt{4 (e_1 - w)(e_2 - w)(e_3 - w)}}.
\]
By~\eqref{eq:eta} and~\eqref{eq:legendre}, the first term is $-i \tfrac{\pi}{\omega_1}$. This gives,
\[
   \forall E \in [-e_3, \infty), \quad  \xi_{\pm}(E) = \mp \int_{e_3}^{-E} \dfrac{(w + \tfrac{\eta_1}{\omega_1}) \rd w}{\sqrt{4 (e_1 - w)(e_2 - w)(e_3 - w)}} \mp \frac{\pi}{\omega_1}.
\]

For $a \in [\tfrac12 \omega_1, \tfrac12 (\omega_1 + \omega_2)]$, that is $E \in [-e_1, -e_2]$, we choose $a_0 = \tfrac12 \omega_2$. This time, we have $\zeta(a_0) - \frac{\eta_1}{\omega_1} a_0 = 0$. Performing a similar reasoning, we obtain
\[
    \forall E \in (-e_1, -e_2), \quad \xi_{\pm}(E) =  \int_{-E}^{e_1} \frac{(w+\frac{\eta_1}{\omega_1})\,\rd w}{\sqrt{4(e_1-w)(w-e_2)(w-e_3)}}.
\]

\noindent\textbf{The density of states.}
We have already shown that the spectrum of $-\partial_x^2+W$ in $L^2(\R)$ is $[-e_1,-e_2]\cup[-e_3,+\infty)$. We now compute the \emph{integrated} density of states, which is equal to $N_W(E) = \pi^{-1} \xi_\pm(E)$ for $E\in\sigma(-\partial_x^2+W)$, where the sign is chosen in such a way that $N_W$ is increasing. Using the above formulas we obtain
\begin{equation*}
N_W(E) =
\begin{cases}
0 & \text{if}\ E \in(-\infty,-e_1] \,,\\
(2\pi)^{-1} \int_{-E}^{e_1} \frac{(w+\frac{\eta_1}{\omega_1})\,\rd w}{\sqrt{(e_1-w)(w-e_2)(w-e_3)}} & \text{if}\ E\in[-e_1,-e_2] \,, \\
\omega_1^{-1} & \text{if}\ E\in[-e_2,-e_3] \,,\\
\omega_1^{-1} + (2\pi)^{-1} \int_{-E}^{e_3} \frac{(w+\frac{\eta_1}{\omega_1})\,\rd w}{\sqrt{(e_1-w)(e_2-w)(e_3-w)}} & \text{if}\ E\in[-e_3,+\infty) \,.
\end{cases}
\end{equation*}
Taking a derivative gives the sought-after density of states
\begin{equation}
n_W(E)=N_W'(E) =
\begin{cases}
(2\pi)^{-1} \frac{-E+\frac{\eta_1}{\omega_1}}{\sqrt{(e_1+E)(-E-e_2)(-E-e_3)}} & \text{if}\ E\in(-e_1,-e_2) \,, \\
(2\pi)^{-1} \frac{E-\frac{\eta_1}{\omega_1}}{\sqrt{(e_1+E)(e_2+E)(e_3+E)}} & \text{if}\ E\in(-e_3,+\infty) \,,\\
0 & \text{otherwise.}
\end{cases}
\label{eq:DOS_W}
\end{equation}


\noindent\textbf{Passing to Jacobi elliptic functions.} We finally relate the potential $W$ with the potential $V_k$ of Theorem~\ref{thm:1D}. From~\cite[8.169.1]{GraRyz}, we have
$$
\wp(z) = e_3 + \frac{e_1-e_3}{\sn(\sqrt{e_1-e_3}z|k)^2}, \quad \text{with} \quad k = \sqrt{\frac{e_2 - e_3}{e_1 - e_3}}.
$$
 Let $0 < k < 1$, and choose
\[
\omega_1 := 2 K(k) \in \R^+ \quad \text{and} \quad \omega_2 := 2 i K'(k) \in i \R^+.
\]
Using $k \sn(u|k)\sn(u+iK'(k)|k) =1$ (see~\cite[10.5.115]{Simon2A}), we get
\[
    W(x) = 2 \wp(x + \tfrac12 \omega_2) = 2 e_3 + 2 (e_1 -e_3) k^2 \sn^2 ( \sqrt{e_1 - e_3} x | k)^2.
\]
It remains to find the values of $e_1$, $e_2$ and $e_3$. First, since $x\mapsto \sn(x|k)^2$ has period $2K(k)$, while $W$ has period $\omega_1$, we deduce that $\sqrt{e_1 - e_3} = 1$. In addition, we always have $e_1 + e_2 + e_3 = 0$, because the coefficient in front of $w^2$ of the polynomial $4w^3-g_2w-g_3$ vanishes. So
\[
    \begin{cases}
        e_1 - e_3 = 1, \\
        k^2(e_1 - e_3) = (e_2 - e_3), \\
        e_1 + e_2 + e_3 = 0,
    \end{cases}
    \quad \text{which gives} \quad
    \begin{cases}
    e_1 = 1 - \tfrac13 (1+k^2), \\
    e_2 = k^2 - \tfrac13 (1+k^2), \\
    e_3 = - \tfrac13 (1 + k^2).
    \end{cases}
\]
Altogether, this proves that
\[
    W(x) = - \tfrac23 (1 + k^2) + 2 k^2 \sn^2(x | k)^2 = V_k(x) + \tfrac13 (1 + k^2).
\]
Since the operator $(-\partial_x^2 + W)$ has spectrum $[-e_1, -e_2] \cup [-e_3, \infty)$, the operator $(-\partial_x^2 + V_k)$ has spectrum $[-1, -k^2] \cup [0, \infty)$, as wanted. The claimed formula~\eqref{eq:nE} for $n(E)$ easily follows from~\eqref{eq:DOS_W}.

%

\begin{acknowledgments}
This project has received funding from the U.S. National Science Foundation (grant agreements DMS-1363432 and DMS-1954995 of R.L.F.) and from the European Research Council (ERC) under the European Union's Horizon 2020 research and innovation programme (grant agreement MDFT 725528 of M.L.).
\end{acknowledgments}


\end{document}